\documentclass[10pt]{article}

\usepackage{amsmath}
\usepackage{amssymb}

\usepackage{graphicx}
\usepackage{epstopdf}
\usepackage{cite}

\newtheorem{definition}{Definition}
\newtheorem{claim}{Theorem}
\newtheorem{proof}{Proof}

\usepackage[labelfont=bf,labelsep=period,justification=raggedright]{caption}

\makeatletter
\renewcommand{\@biblabel}[1]{\quad#1.}
\makeatother

\date{}

\begin{document}

\begin{flushleft}
{\Large
\textbf{The Fiber Walk: A Model of Tip-Driven Growth with Lateral Expansion}
}
\\
Alexander Bucksch$^{(1,2)}$, 
Greg Turk$^{(1)}$, 
Joshua S. Weitz$^{(2,3)}$
\\
\bf{1} Georgia Institute of Technology, School of Interactive Computing, Atlanta, GA, USA
\\
\bf{2} Georgia Institute of Technology, School of Biology, Atlanta, GA, USA
\\
\bf{3} Georgia Institute of Technology, School of Physics, Atlanta, GA, USA
\\
$\ast$ E-mail: Corresponding bucksch@gatech.edu
\end{flushleft}

\section*{Abstract}

Tip-driven growth processes underlie the development of many plants. To
date, tip-driven growth processes have been modeled as an elongating path
or series of segments, without taking into account lateral expansion during
elongation. Instead, models of growth often introduce an explicit
thickness by expanding the area around the completed elongated path.
Modeling expansion in this way can lead to contradictions in the physical
plausibility of the resulting surface and to uncertainty about how the
object reached certain regions of space. Here, we introduce \emph{fiber
walks} as a self-avoiding random walk model for tip-driven growth
processes that includes lateral expansion. In 2D, the fiber walk takes
place on a square lattice and the space occupied by the fiber is modeled
as a lateral contraction of the lattice. This contraction influences the
possible subsequent steps of the fiber walk. The boundary of the area
consumed by the contraction is derived as the dual of the lattice faces
adjacent to the fiber. We show that fiber walks generate fibers that have
well-defined curvatures, and thus enable the identification of the process
underlying the occupancy of physical space. Hence, fiber walks provide a
base from which to model both the extension and expansion of physical
biological objects with finite thickness.

\section*{Introduction}

The growth and development of biological organisms involves processes at
multiple scales, including regulation of the timing and location of
specific developmental programs. Despite their structural complexity, many
prototypical processes involved in development have been identified. One
such process is the notion of tip-driven growth in plants \cite{rounds2013} that leads
to the successive elongation of plant branches or roots below ground. Plant growth also involves the process of expansion, e.g., the thickening
of plant roots. The result of such developmental programs in plants are extended and
thickened structures, e.g., complex roots and crowns.
\begin{figure}[ht]
\centering
\includegraphics[scale=0.2]{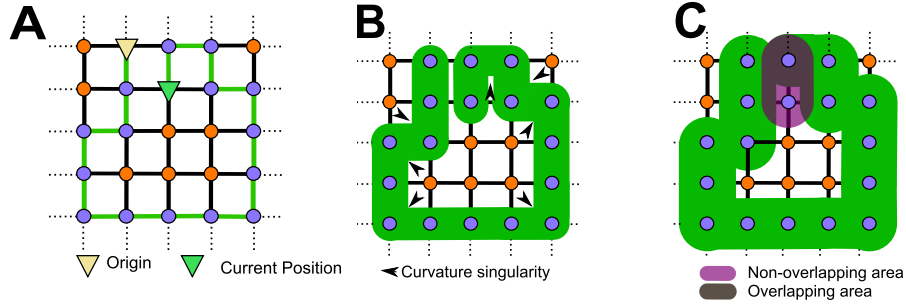}
\caption{Example of contradictions arising from decoupling elongation and expansion of a SAW. (A) SAW of 16 steps beginning at the origin (yellow triangle) with current position marked via the green triangle. (B) Thickened SAW - note the curvature singularity of the resulting surface. (C) SAW thickened by $2/3$ of an edge length - multiple regions in space are ``associated'' with different steps, leading to an identifiability problem, which holds for all thickening greater than or equal to $1/2$ of an edge length.}
\label{fig:motivation}
\end{figure}
A number of models have been proposed to examine the growth of extended structures. These vary in complexity. For example, perhaps the simplest, from a parameterization perspective, is the self-avoiding random walk (SAW) \cite{madras1993,landau2009}. A SAW is a sequence of steps along the line segments of grid edges that connect the locations of crossing grid-lines vertices. Such a grid is a square lattice in $\mathbb{Z}^2$ and the SAW satisfies the constraint that no vertex is visited twice and all walks are equally likely.

In SAWs and derived models, the elongation and expansion steps are decoupled. The consequences of such decoupling is that the resulting shape may be non-physical. In Fig.~\ref{fig:motivation} we show a SAW thickened to $1/3$ and $2/3$ of an edge length. In both cases, the SAW can be associated
with a boundary which is exactly the expansion distance away from any point
along the edges in the SAW. The derived boundary has singular points.
Singular points are associated with vanished curvature, which are
problematic from both geometric and biophysical perspectives. Next,
in the case of an expansion exceeding $1/2$ of the edge length,
the walk may intersect with itself, thus the resulting boundary is not a manifold, because no location can be occupied by two boundary points. Hence, decoupling elongation and growth raises the possibility of models generating an elongated object that is feasible only in the limit of infinitesimal spatial expansion. This violation of the manifold property leads often to difficulties in current reconstruction practice from random walks \cite{karch1999,hamarneh2010} and discrete samples in general \cite{bhattacharya2013,dey2013}.

Inspired by the growth of the taproot in plants in particular and by
growth processes in biology in general, we introduce here the fiber walk:
a SAW model that includes a notion of lateral expansion of the path. We
consider the simplest case of a randomly growing line that thickens while
elongating into a random direction. We examine the fiber walk in 2D and 3D, and note that 
although plant roots self-evidently grow in 3D, a number of experimental
studies restrict root growth to 2D or quasi-2D (e.g., \cite{Hund2009}). Intuitively, a lattice defines a space
in which the fiber walk grows a fiber along the edges of the lattice. The
interaction of the fiber walk with the lattice is two-fold: 1) the
elongation of the walk is achieved by stepping to unvisited adjacent
vertices on the lattice, and 2) the spatial expansion is accomplished by
contracting certain unvisited vertices adjacent to a visited vertex from
the side. In comparison to the ordinary SAW, our new contraction allows us
to define a discrete approximation of the area that is reserved for the
fiber to expand. Such an expansion implies a boundary, which is a curve in 2D
or a surface in 3D defining an area that gives insight into regions of the
lattice where no elongation and/or expansion can occur. Such inaccessible regions are caused by self-avoidance that results from the elongation and expansion of the growing
object. As we will demonstrate, the fiber walk resolves the contradictions
of decoupled growth models while creating a base for future studies of the
relationship between growth and form.

\section*{Method}
\subsection*{Overview}
We introduce the formal notion of the fiber in $n$ dimensions, whose
process is a walk over the edges of a lattice. First we introduce the
lattice as a graph consisting of vertices connected by edges. Secondly, we
define the fiber as a subset of all possible graphs on the a lattice whose
vertices have at most two incident edges. The fiber is characterized by
the presence of a stopping configuration on the lattice that we define.
After introducing the two basic notions of lattice and fiber walk, we
define fiber walks as the process that creates the fiber as sequence of
steps along lattice edges. The interaction between the fiber walk and the
lattice extends the class of SAWs. The extension adds the local
contraction of the lattice around the last reached position, which in turn
is the space consuming expansion of the fiber walk. We illustrate the
notions with 2D examples.

\begin{figure}[ht]
\centering
\includegraphics[scale=0.5]{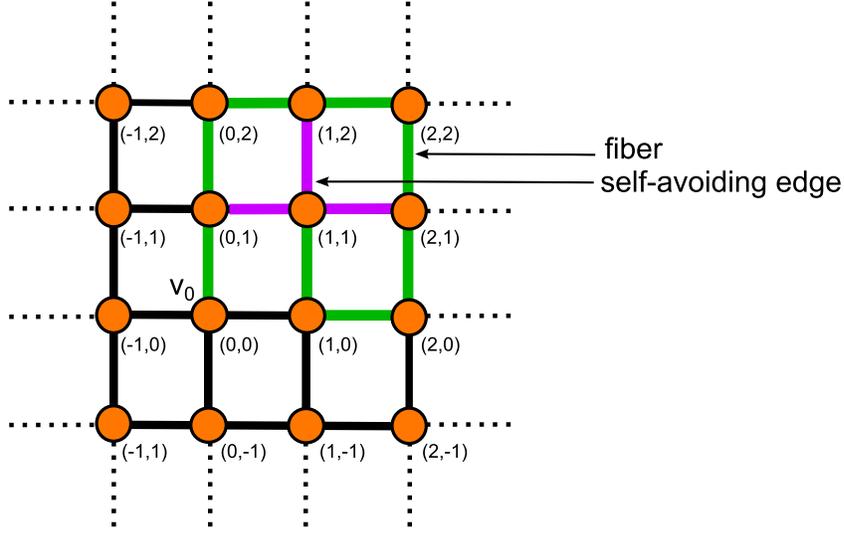}
\caption{A growing fiber (green) and its self-avoiding edges (purple). The example shows a fiber on a lattice as defined in Def.~\ref{position} and the vertex positions are assigned to each vertex.}
\label{fig:SA}
\end{figure}

\subsection*{The lattice}
\label{notion}

A lattice in $\mathbb{Z}^n$, on which the fiber walk grows a fiber, is a graph $L(V,E)$ of infinite size consisting of a set of vertices $V$ and a set of edges $E$. Such a lattice is derived by taking the graph Cartesian product of $n$ one-dimensional lattices at infinite length in $\mathbb{Z}$ , where $n$ denotes the dimension of the targeted lattice \cite{skiena1991}. Beside the notion of dimension, we define an embedding for the lattice by assigning unit edge length to every edge of the lattice. A position is assigned to every vertex, such that the vertex positions are denoted as $n$-tuples in $\mathbb{Z}^n$ (see Fig.~\ref{fig:SA}). The assignment of positions to each vertex will enable us later to define contractible edges.

\begin{definition}
\label{position}
Let $L$ be a lattice in $\mathbb{Z}^n$. Fix $v_0$ as the origin and assign unit length to all edges, then the position of each vertex is the $n$-tuple associated with $\mathbb{Z}^n$ .
\end{definition}

\subsection*{The fiber}

We aim at defining the fiber as a non-branching and loop free graph $F(V,E)$ consisting of set of vertices $V$ and a set of edges $E$ \cite{wilson1985}, growing on an infinite lattice whose vertices are enumerated in $\mathbb{Z}^n$. The fiber starts at a chosen vertex, the so-called origin. We start with a narrow definition of a directed path graph, excluding the trivial case of an `empty' path containing no edges. Furthermore, we restrict ourselves in the context of this paper to only one fiber on the lattice.
\begin{definition}
\label{def:directed}
Let $G$ be a graph with $m+1, 1\leq m \leq \infty$ vertices, such that exactly two vertices have one incident edge and $m-1$ vertices have exactly 2 incident edges. Fix $v_0$ and $v_m$ to be the vertices with one incident edge. Such a graph is called a \textit{directed path graph}. Also it is said to be \textit{directed by} the sequence $\{v_0,v_1,...,v_{m-1},v_{m}\}$ 
\end{definition}

In the following two definitions we introduce the fiber as a subset of possible directed path graphs. First we define self-avoiding edges to limit ourself to cycle-free directed path graphs and as a measurable characteristic of self-avoidance of the walk. Secondly, we use self-avoiding edges to obtain a stopping criteria of the fiber.
\begin{definition}
\label{def:SA}
Let $L$ be a lattice and $G$ be a directed path graph on $L$. An edge e belonging to the edge set of L is associated with two vertices belonging to the vertex set of $G$, is called a self-avoiding edge if $e$ is not in the edge set of $G$.
\end{definition}
Given self-avoiding edges we can define the stopping configuration of a fiber as a vertex that is only incident to self-avoiding edges.

\begin{definition}
\label{def:fiber}
Let $L$ be a lattice. Any directed path graph $G=\{v_0,v_1,...,v_{m}\}, 1\leq m \leq \infty$ is called a fiber if and only if all $v_i$ are distinct. A fiber is said to be stopped if and only if all incident edges of $v_{m}$ are self-avoiding edges.
\end{definition}

\begin{figure}[ht]
\centering
\includegraphics[scale=0.4]{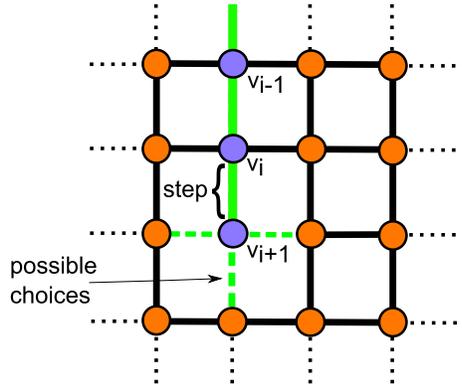}
\caption{Random choice of a follow up step. A step (green) shown on a 2D lattice. The possible follow-up steps are represented as a dotted line.}
\label{fig:StepandRandomChoice}
\end{figure}
\subsection*{The fiber walk}
In the following we describe the process that constructs the fiber. The process of constructing a fiber is called a fiber walk. The fiber walk is a sequence of steps, each from a vertex $v_i$ to an adjacent vertex $v_{i+1}$ via an edge, that is taken by an equally random choice of all incident edges to $v_{i}$ to reach $v_{i+1}$. Consequently, at the newly reached vertex $v_{i+1}$ a new random choice is executed. The choice is taken under the constraint that the newly chosen edge does not connect to a previously visited vertex. The excluded choices prevent the fiber walk from walking back on itself and forming loops. The process is depicted in Fig.~\ref{fig:StepandRandomChoice}.

\begin{definition}
\label{def:step}
Let $v_i$ be a vertex that belongs to a fiber with at least 2 vertices in its vertex set. One step from a vertex $v_i$ to a vertex $v_{i+1}$ along an edge of a lattice is the equal random choice among edges incident to $v_i$ that are not incident to any $v_{i-k}, 1\leq k \leq i$.
\end{definition}

Def.~\ref{def:fiber} -\ref{def:step} are equivalent to the notion of a growing self-avoiding random walk (Fig.~\ref{fig:StepandRandomChoice}) on a lattice \cite{lyklema1984}. It is sufficient for our purpose to investigate the growing SAW, because its configurations are possible configurations of the traditional SAW used to study the excluded volume effect.
\begin{figure*}
\centering
\includegraphics[scale=0.4]{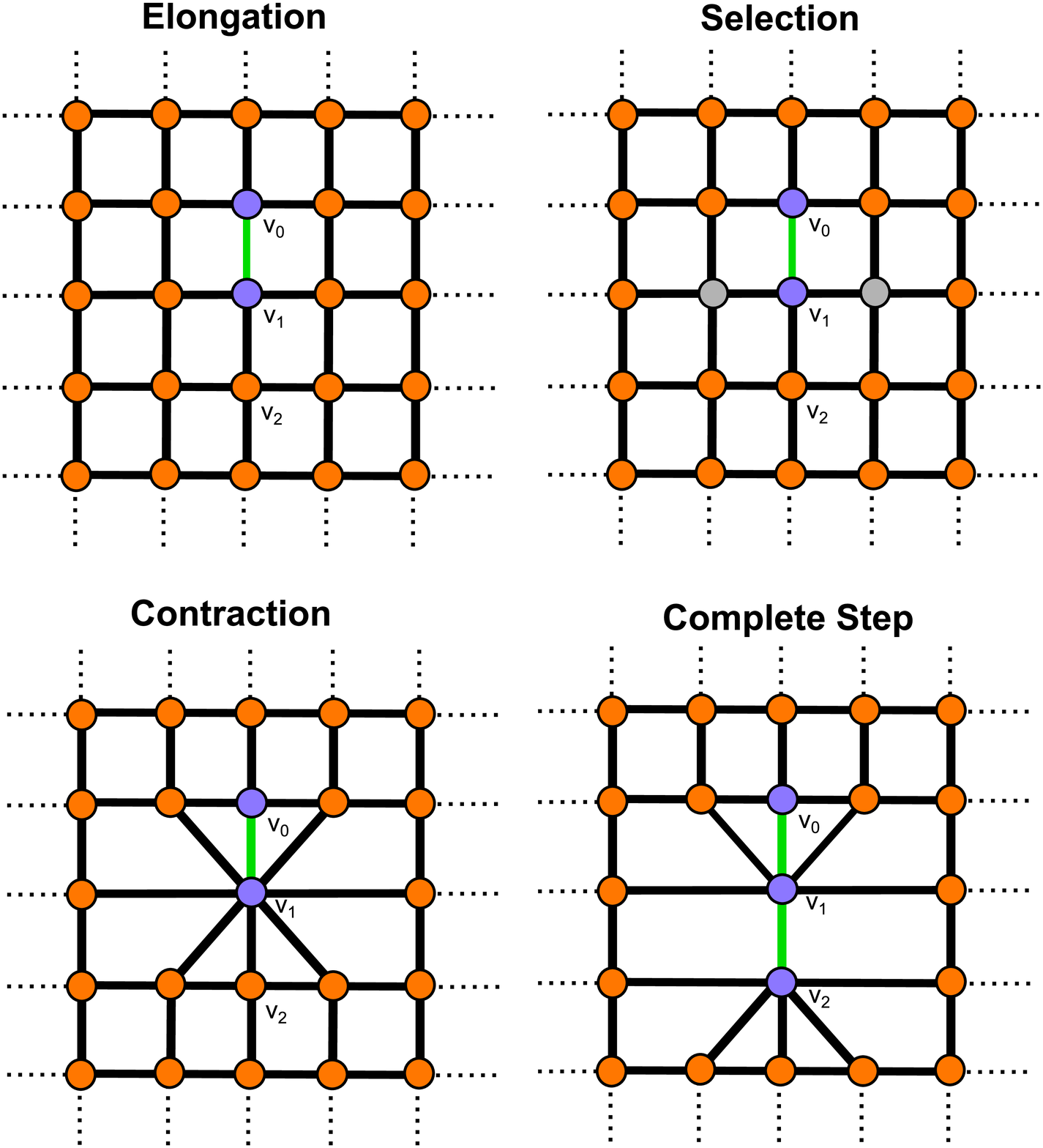}
\caption{Fiber walk on a 2D lattice. Elongation occurs as the first step which is chosen randomly between the edges incident to the origin $v_0$ (compare Fig.~\ref{fig:StepandRandomChoice}). Here the chosen edge to reach $v_1$ is shown in green. Selection of vertices incident to the walk from the side correspond to the expansion. The vertices selected to be merged with $v_1$ are shown in grey. Contraction of the selected vertices and its result after merging the selected vertices to $v_1$. Another step, including elongation and expansion, is shown as a second step on the lattice. The second step reaching $v_2$ uses the same color schema as before.}
\label{fig:Contraction}
\end{figure*}
\subsection*{Fiber walk contraction of the lattice}
\label{contraction}
In our fiber walk each elongation step causes a contraction on the lattice, which is a prerequisite to reconstruct the expanded boundary locally at each time step. The contraction of the lattice at each step is performed as a set of merges between vertices adjacent to the last reached vertex on the lattice. A merge is the union of two adjacent vertices $v_i$ and $v_j$, where $v_i$ inherits all edges of $v_j$. A complete contraction at $v_i$ is given as a set of merging operations driven by an edge labeling. We first define the merging operation of two vertices. We denote the merge between two vertices as $\oplus$. The process is depicted in Fig.~\ref{fig:Contraction}.

\begin{definition}
\label{def::contr}
Let $v_1$ be the last vertex reached on the lattice of a fiber walk and let $v_2$ be a vertex connected by a non-self-avoiding edge $e$ to $v_1$. The merge $v_1 \oplus v_2$ results in $v_{1}$ inheriting all edges incident to $v_2$ except for $e$, which is discarded.
\end{definition}

Each merge is said to cause a new edge to exist on the lattice, because two previously distant vertices are newly adjacent after each merge. Therefore, a number of merges are associated with each edge. 

In the following we define an initial edge labeling based on the vertex positions introduced in Def.~\ref{position}. After that we give the mechanism to recalculate labels after each contraction step.

\begin{definition}
\label{def:lables}
Let $v_1$ and $v_{2}$ be two vertices on a lattice with positions $p_1$ and $p_2$ and let both vertices be connected by an edge. The edge connecting $v_1$ and $v_2$ is said to be bi-directed with labels $(p_1-p_2)$ denoting the direction of the edge between $v_1$ and $v_2$ and $(p_2-p_1)$ denoting the direction of the edge between $v_2$ and $v_1$.
\end{definition}
For example, the labels for all 3D-directions are: $(\pm 1,0,0)$, $(0,\pm 1,0)$,$(0,0,\pm 1)$. The labels of edges incident to a vertex of a lattice contain only entries of -1,0 and 1, because we initially assigned unit length to all edges. The characteristic of the bidirectional labeling is that a vertex $v_i$ has an outgoing edge to all adjacent vertices $v_{i+1}$ and an incoming edge from each $v_{i+1}$ to $v_i$.

The recalculation of the labels considers the merge of two vertices $v_i \oplus v_j$ and is denoted as $\boxplus$. Let $l_1$ be the label belonging to the edge $e_1$ connecting $v_i$ to $v_j$. Let the label $l_2$ belong to an edge $e_2 \neq e_1$ incident to $v_j$. The label of a newly created edge is then calculated for each entry $a$ in the Cartesian $n$-tuple of the label, such that:
\begin{eqnarray}
a\boxplus a = a \\
a \boxplus (-a) = (-a) \boxplus (a) = 0 \\
a \boxplus 0 = 0 \boxplus a = a
\end{eqnarray}

Equations 1-3 assure that the entries of the direction labels are always -1,0 or 1. For example, if $l_1=(1,0,-1)$ and $l_2=(1,1,-1)$, then $l_1 \boxplus l_2 = (1,1,-1)$. 
\pagebreak

We identify edges for contraction by an entry in the $n$-tuple of their associated labels. The selection of edges to contract is given by two rules: 
A contraction is performed as a set of merging operations with vertices having an incident outgoing edge $e$ to $v$, iff:\\
1. $e$ is an outgoing edge of $v$, that differs in its label in more than one value with the last incoming edge to $v$ added to the fiber.\\
2. $e$ is not a self-avoiding edge\\

Edges selected for contraction, are said to be incident from the side and reflect the directions of lateral expansion. In practice, we select all edges that do not have an entry in their edge label that indicates a movement opposite to the current growth direction and is different from the elongation step taken in the current growth direction. Recall that edges belonging to the edge set of the fiber are self-avoiding edges.

\begin{figure}[ht]
\centering
\includegraphics[scale=0.35]{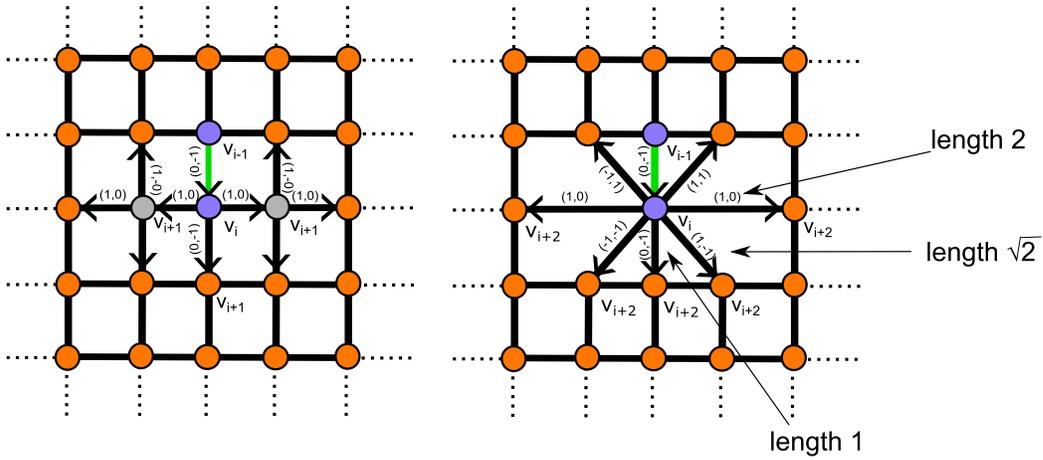}
\caption{ The lattice contraction of the fiber walk. Left: The fiber (green) on a lattice and its edges selected for contraction. Right: The contracted edges which are possible follow up steps of length 1,$\sqrt{2}$ and 2. In both figures the edge labels involved in the contraction and their direction, indicated as arrows are shown.}
\label{fig:pcase1}
\end{figure}

The contraction results in edges of different lengths in the edge set of the fiber and of self-avoiding edges. In the following we show that a fiber has three edge length classes and self-avoiding edges have five edge length classes in 2D by investigating the edge set of the lattice, but the proofing scheme applies to higher dimensions as well. We make use of two prerequisites; first we consider two cases of merging results: (1) merges not resulting in self-avoiding edges (Fig. \ref{fig:pcase1}) and (2) merges resulting in self-avoiding edges (Fig. \ref{fig:pcase2}). Both cases are distinguishable, because from Def.~\ref{def::contr} it follows that a contracted edge on the lattice has one vertex in the vertex set of the fiber and one vertex exclusively in the vertex set of the lattice. Secondly, we will distinguish between edges with identical edge label that add up their length if their common vertex is merged and edges with distinct edge label, whose resulting edge length is calculated by the Pythagorean theorem. The edge length calculation follows from the definition of the lattice in Def.~\ref{position} and the definition of the edge labels Def.~\ref{def:lables}.
\newpage
\begin{claim}
The edge set of the fiber has three edge length classes.
\end{claim}

\begin{proof}
Let $v_i$ be the last vertex reached by the fiber on the lattice and $v_{i-1}$ be the previously reached vertex on the lattice connected by the edge $e_{v_{i-1},v_{i}}$ with label $l_{v_{i-1},v_{i}}$. Furthermore, let $\{v_{i+1}\}$ be the set of vertices incident to $v_i$ and $\{v_{i+2}\}$ be the set of vertices reachable after contraction via $\{e_{v_{i},v_{i+2}}\}$ with corresponding labels $\{l_{v_{i},v_{i+2}}\}$.\\
\\
$v_i\oplus v_{i+1}$ is performed if $l_{v_{i-1},v_{i}} \boxplus l_{v_{i},v_{i+1}} \neq l_{v_{i-1},v_{i}}$, from which follows, that \\
if $l_{v_{i-1},v_{i}} \boxplus l_{v_{i},v_{i+1}} = l_{v_{i-1},v_{i}}$ no merge is performed, $e_{v_{i},v_{i+1}} = e_{v_{i},v_{i+2}}$  with unit length of 1. \\
If $l_{v_{i},v_{i+1}} \neq l_{v_{i-1},v_{i}}$ and $l_{v_{i},v_{i+1}} = l_{v_{i+1},v_{i+2}}$,  $e_{v_{i},v_{i+2}}$ has length $1+1=2$.\\
If $l_{v_{i},v_{i+1}} \neq l_{v_{i+1},v_{i+2}}$ and  $l_{v_{i},v_{i+1}} \neq l_{v_{i+1},v_{i+2}}$, then $e_{v_{i},v_{i+2}} $ has length $\sqrt{1^2+1^2}=\sqrt{2}$.
\begin{flushright}
$\blacksquare$
\end{flushright}
\end{proof}
\begin{figure}[ht]
\centering
\includegraphics[scale=0.27]{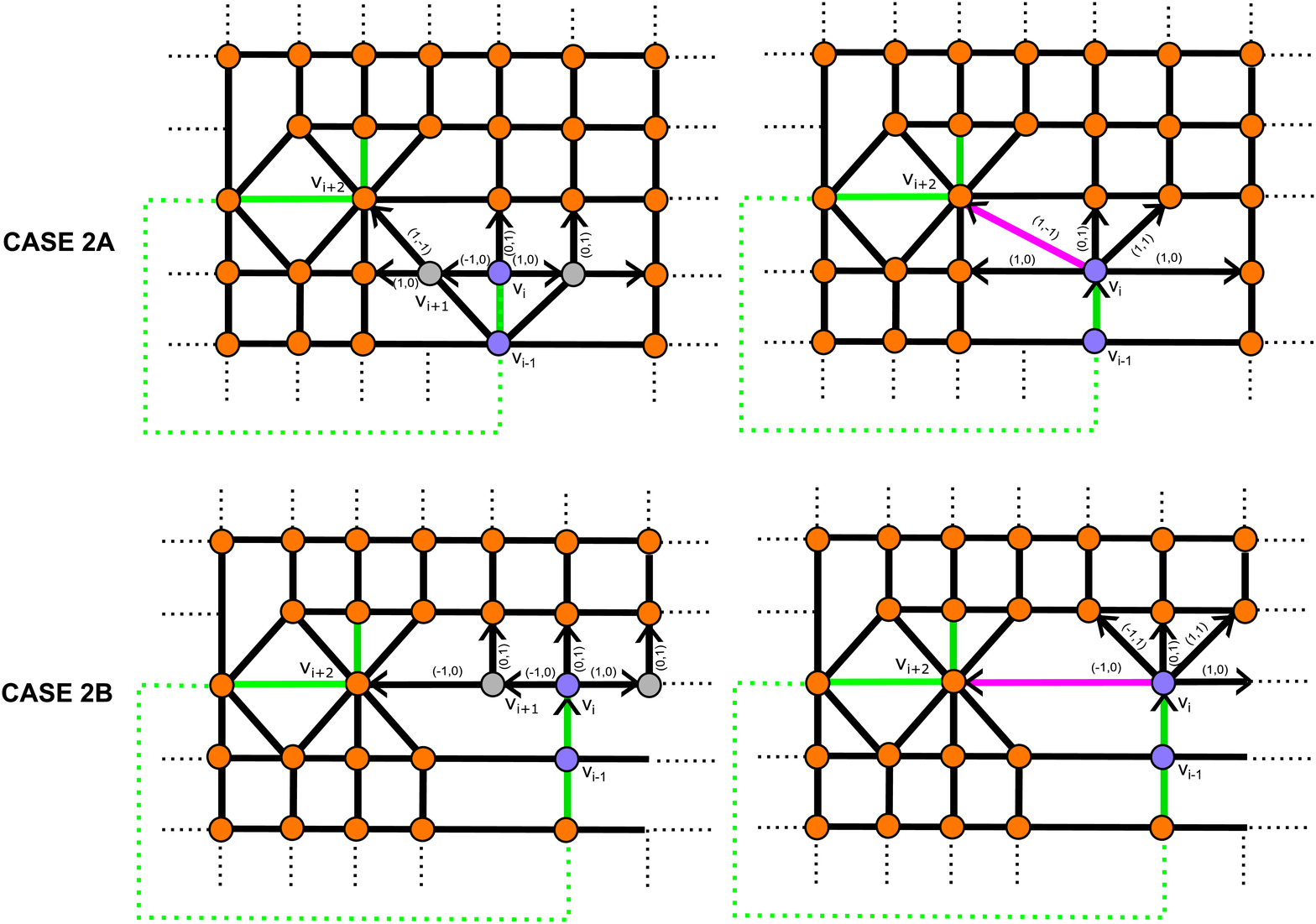}
\caption{Edge length classes of the fiber walk. Case 2a shows a self-avoiding edge (purple) of length $\sqrt{3}$ and Case 2b shows a self-avoiding edge of length 3. The fiber on a lattice is colored green and its edges selected for contraction are shown in grey. The dotted green line denotes an unknown fiber walk that is not affecting the given configuration. In both figures the edge labels involved in the contraction and their direction indicated as arrows are shown.}
\label{fig:pcase2}
\end{figure}
Proof 1 exhausts all combinations of the introduced prerequisites. It is trivial, that each of the three edge length classes of the fiber are possible edge length classes of self-avoiding edges, because the fiber walk has a non-contracting edge that can connect to each of the edge length classes of the fiber and to the fiber directly. This occurs if $v_{i+1}$ belongs to the vertex set of the fiber and no merge is performed. The edge length is therefore 1,2 or $\sqrt{2}$. The edge length classes of the fiber were shown as possible follow up steps. The follow-up step itself is a random choice and therefore only one of the possible follow-up steps are chosen. Hence, the edge length classes also apply to non-self-avoiding edges incident to the walk, which may become self-avoiding. In the following we use the edge length classes of Proof 1 as a starting point to exhaust all possible 2D cases of self-avoiding edge lengths.

\begin{claim}
Self-avoiding edges have five edge length classes.
\end{claim}

\begin{proof}
Let $v_i$ be the last vertex reached by the fiber on the lattice and $v_{i-1}$ be the previously reached vertex on the lattice connected by the edge $e_{v_{i-1},v_{i}}$ with label $l_{v_{i-1},v_{i}}$. Furthermore, let $\{v_{i+1}\}$ be the set of vertices incident to $v_i$ and $\{v_{i+2}\}$ be the set of vertices reachable after contraction via $e_{v_{i},v_{i+2}}$ with label $l_{v_{i},v_{i+2}}$. We consider the cases of a self-avoiding edge to be formed by $v_i\oplus v_{i+1}$ from the configuration of $e_{v_{i},v_{i+1}}$ having a vertex $v_{i+2}$ incident to $v_{i+1}$ that belongs to the vertex set of the fiber.\\
\\
If $e_{v_{i+1},v_{i+2}}$ has length $\sqrt{2}$ and $l_{v_{i+1},v_{i+2}} \neq l_{v_{i},v_{i+1}}$ it follows that $e_{v_{i},v_{i+2}} $ has length$\sqrt{\sqrt{2}^2+1^2}=\sqrt{3}$.\\
If $e_{v_{i+1},v_{i+2}}$ has length $2$ and $l_{v_{i+1},v_{i+2}} \neq l_{v_{i},v_{i+1}}$ it follows that $e_{v_{i},v_{i+2}}$ is of length ${\sqrt{2^2+1^2}=\sqrt{5}}$.\\
If $e_{v_{i+1},v_{i+2}}$ has length $2$ and $l_{v_{i+1},v_{i+2}} = l_{v_{i},v_{i+1}}$ it follows that $e_{v_{i},v_{i+2}}$ is of length $2+1=3$.\\
If $e_{v_{i+1},v_{i+2}}$ has length 1 and $l_{v_{i+1},v_{i+2}} \neq l_{v_{i},v_{i+1}}$, it follows that $e_{v_{i},v_{i+2}}$ is of length  $\sqrt{2}$.\\
If $e_{v_{i+1},v_{i+2}}$ has length 1 and $l_{v_{i+1},v_{i+2}} = l_{v_{i},v_{i+1}}$, it follows that $e_{v_{i},v_{i+2}}$ is of length $2$.

\begin{flushright}
$\blacksquare$
\end{flushright}
\end{proof}

\subsection*{The expansion of the fiber walk boundary}
\label{manifold}
In the following we note that a face is bounded by edges that connect vertices. Our central interest is to define the spatial expansion of the fiber walk as represented by a boundary. We introduce the fiber walk boundary for simplicity in 2D, but the principle extends naturally to higher dimensions. Based on Def.~\ref{def:SA}, we can distinguish three kinds of faces, which share an edge or a vertex with the walk: 1.) \emph{self-avoiding faces}, whose bounding edge set contains only edges that are self-avoiding or belonging to the fiber walk, 2.) \emph{non-self-avoiding faces}, having no self-avoidend edges in their bounding edge set or belong to the fiber walk and 3.) \emph{mixed faces}, bounded by edges that are self-avoiding, not self-avoiding or belonging to the fiber walk.

We first introduce the basic definition of the boundary by assuming that no self-avoiding edges are present on the lattice. In a second step we will extend this notion to self-avoiding edges.Fig.~\ref{fig:boundary} shows the derived boundary of the fiber walk on the minimal example given before in Fig.~\ref{fig:Contraction}.
\begin{figure}[ht]
\centering
\includegraphics[scale=0.3]{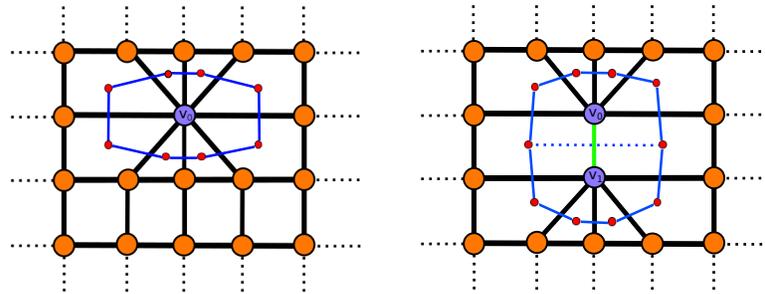}
\caption{The fiber walk boundary. The boundary (blue) of the fiber walk is derived from the face dual of the lattice (red) in 2D. The shown configuration corresponds to the example given in Fig.~\ref{fig:Contraction}. Here, dotted line segments denote non-unique edges, which do not belong to the boundary.}
\label{fig:boundary}
\end{figure}

\begin{definition}
\label{def:dual1}
Let $L$ be a lattice, then the dual lattice $L_d$ of $L$ has vertices $v$ each of which corresponds to a face of $L$ and each of whose faces corresponds to a vertex of $L$. Each vertex $v_i$ is connected via an edge to the vertex $v_{i+1}$ if the corresponding faces of $L$ are adjacent.
\end{definition}

As a final component of our setting we extend Def.~\ref{def:dual1} to achieve validity in the presence of self-avoiding edges based on an intermediate lattice. The intermediate lattice defines the faces adjacent to the walk by placing vertices at a certain fraction of the edge length of incident edges to the walk. This essential definition, as we see later on, is depicted in Fig.~\ref{fig:surfaceFig}.
\begin{definition}
\label{def:dist}
Let $L$ be a lattice and $F$ be a fiber on the lattice. Additionally, let $E$ be the set of lattice edges $e_i$ incident to the walk and belonging to the same face and $c$ be the number of contractions causing $e_i$ to exist. The intermediate lattice has vertices at maximal distance $d(e_i)=\frac{1}{2\times c}$ along $e_i$.
\end{definition}

Def.~\ref{def:dist} defines a maximal distance to compensate for self-avoiding edges of length $\sqrt{5}$ and 3, that are the result of two merges; here the compensation is defined in terms of $c < 1$. Without this compensation the fiber would expand at previously reached locations whenever a self-avoiding edge of length $\sqrt{5}$ and 3 is formed. These two self-avoiding edge length classes result in a distance between the fiber boundary, which is less than needed for a fiber to grow in-between and expand.

\begin{figure}[ht]
\centering
\includegraphics[scale=0.5]{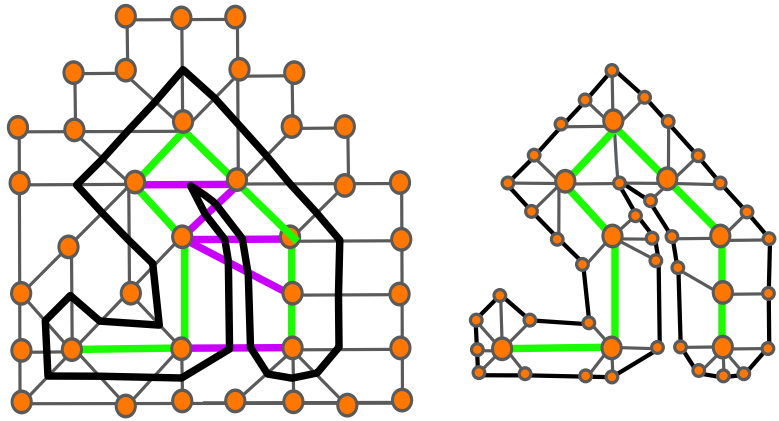}
\caption{Recovering the surface. (left) The original lattice with grey edges and orange vertices. The walk is shown in green and self-avoidend edges are shown in purple. The black line denotes the half-edge length of edges incident to the walk. (right) The intermediate lattice consisting of the black half edge line and the grey edges incident to the green walk. Small orange vertices are placed at half-edge distance while the bigger orange vertices are the original vertices belonging to the walk.}
\label{fig:surfaceFig}
\end{figure}
Def.~\ref{def:dist} refines the lattice such that Def.~\ref{def:dual1} applies to all faces incident to the fiber. We finalize this section with actual computations of the introduced fiber walk in 2D and 3D (Fig.~\ref{fig:Walks}). An example of a computed SAW and the fiber walk in 2D and 3D is given in Fig.~\ref{fig:Walks}. Fig.~\ref{fig:Walks} (a) and (b) show a SAW on a lattice including the self-avoiding edges. In comparison, Fig.~\ref{fig:Walks} (c) and (d) show a fiber walk and the lattice resulting from the contraction. We note that analysis of transient dynamics involve the following caveat: direct comparisons of time-steps (and corresponding statistics) with and without self-avoidance may depend on details of the growth process.
\begin{figure*}
\centering
\includegraphics[scale=0.21]{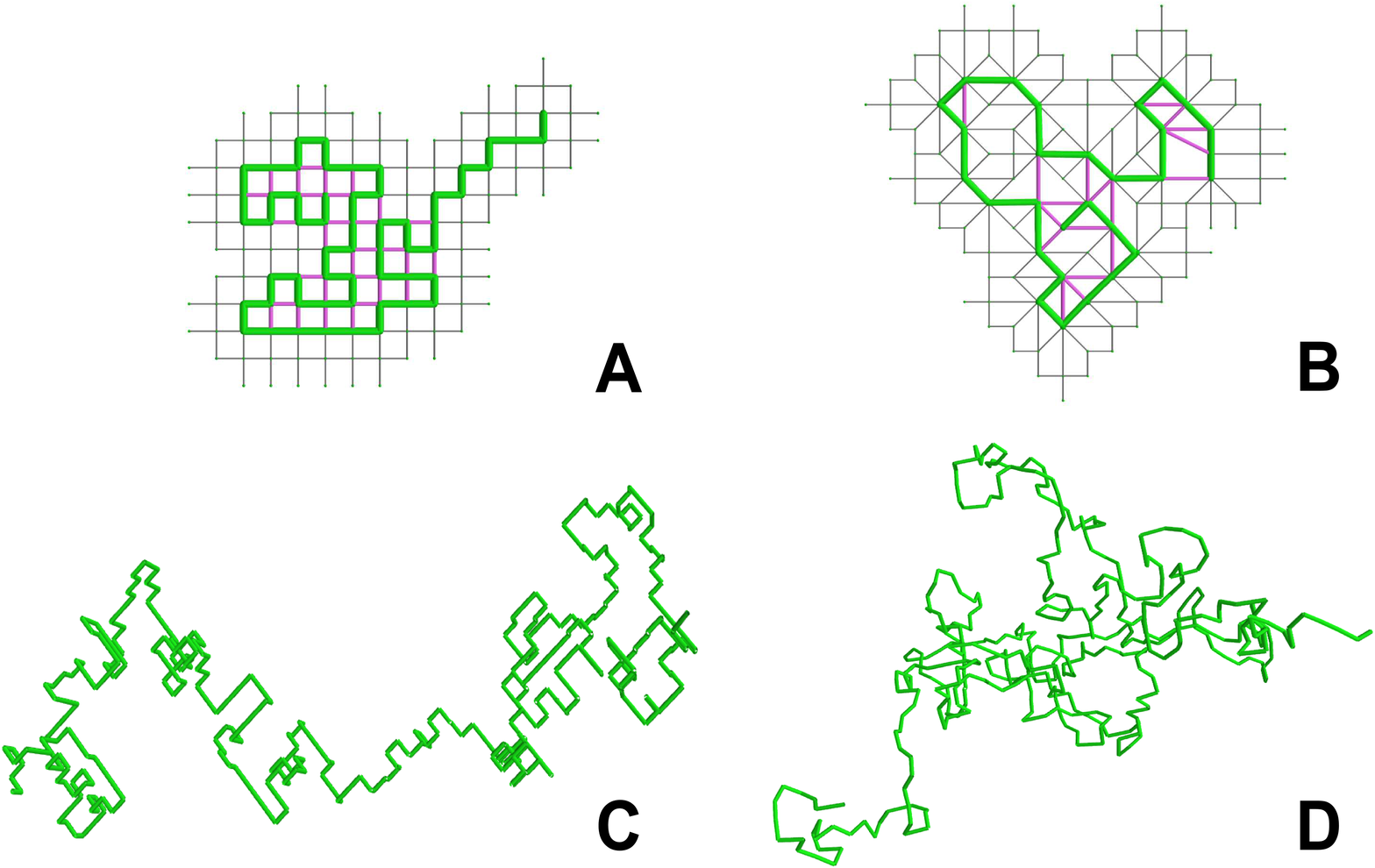}
\caption{Computed comparison of SAW and fiber walk. (A) A growing SAW in 2D and (B) a Fiber Walk in 2D. All walk edges are colored in green, the lattice is shown in (grey) and the self avoiding edges are colored in purple for all images. (C) A growing SAW in 2D and (D) a fiber walk 3D.}
\label{fig:Walks}
\end{figure*}
\section*{Results}
\subsection*{Fiber boundary does not have curvature singularities}
\label{sub:sing}

In the introduction we gave 2D examples of singularities on the boundary
if only the symmetric expansion around the vertices of the walk is
considered. It is trivial to observe that each vertex contributes to the
boundary wherever a vertex of the fiber has an incident edge belonging to
the edge set of the lattice. On an uncontracted lattice, a vertex does not
contribute to the boundary on both sides of the walk at locations where
the walk turns $90^\circ$. We compare here the case of a right angle turn
for the fiber walk to demonstrate the influence of the intermediate
lattice. Fig. \ref{fig:pRightAngle} shows the four possible local
configurations of a fiber walk that cause a right angle in 2D. We can
distinguish two configurations of right angles within the four cases:
configurations containing self-avoiding edges, and configurations
containing no self-avoiding edges.
\begin{figure}[ht]
\centering
\includegraphics[scale=0.29]{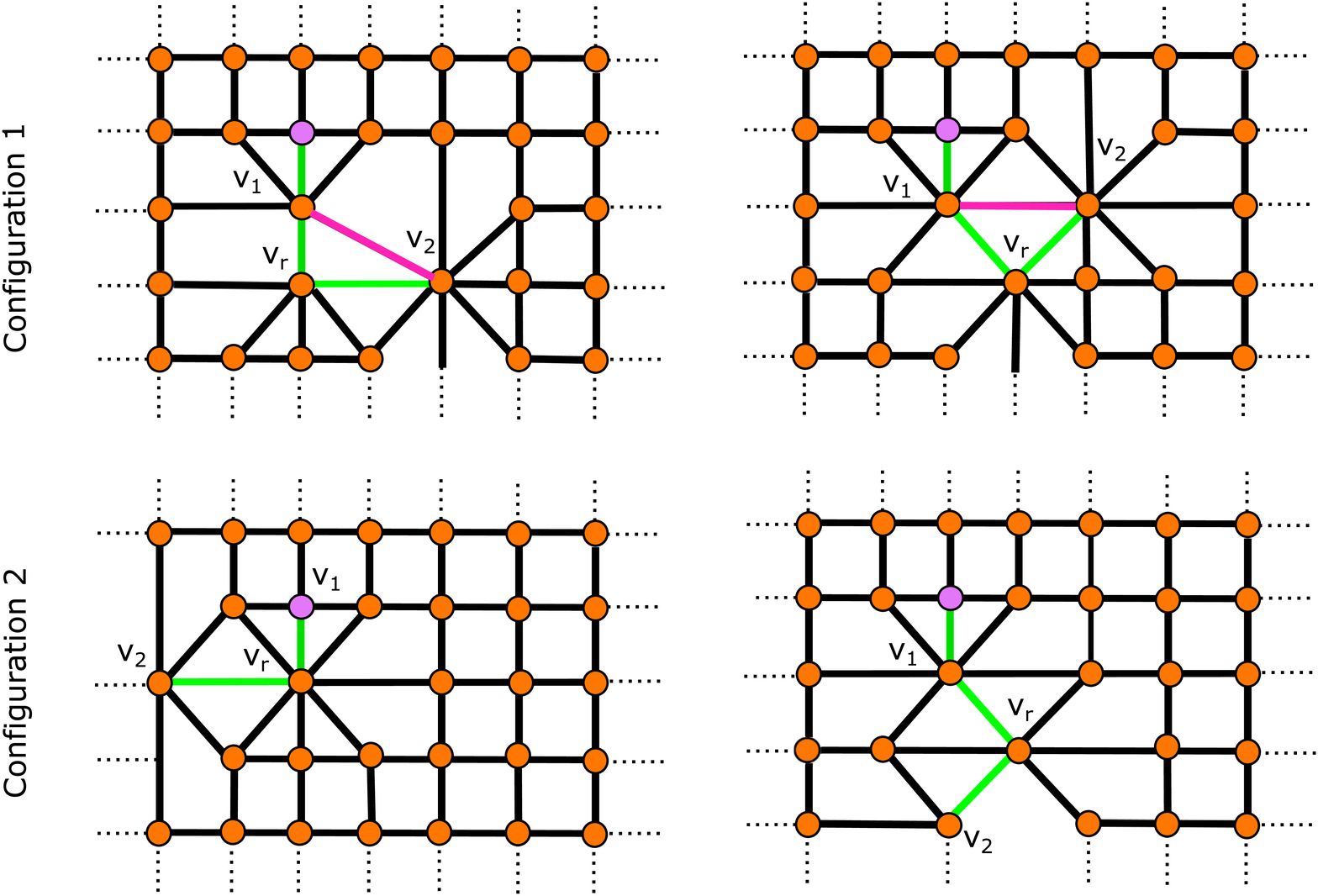}
\caption{Avoiding right angles. The four cases of right angle configurations of the fiber walk. Configuration 1 shows possible right angle configurations containing a self-avoiding edge. Configuration 2 shows the possible configurations of right angles without self-avoiding edges.}
\label{fig:pRightAngle}
\end{figure}

Let $v_r$ be a vertex at which the two incident walk edges connect to vertices $v_1$ and $v_2$ respectively.

\textbf{Configuration 1} has one edge connecting $v_1$ and $v_2$ at the side of the right angle. Hence, Def.~\ref{def:dist} applies for reconstructing the boundary, where the intermediate lattice forms triangular faces having $v_r$ in their vertex set at the side of the right angle.

\textbf{Configuration 2} has an edge $e$ incident to $v_r$ at the side of
the right angle that does not belong to the fiber and connects to a
vertex that exclusively belongs to the lattice. Hence, Def.~\ref{def:dual1}
applies for reconstructing the boundary. The edge $e$ guarantees that $v_r$
contributes to the boundary on the side of the right angle.

Configuration 1 demonstrates the case where our boundary construction
guarantees that the expansion around $v_r$ contributes to the boundary at
the side of the right angle. Configuration 2 is the standard case
resulting from the contraction. The two configurations are shown for a
computed fiber walk in Fig.~\ref{fig:boundarComp}A, along with the
constructed intermediate lattice in Fig.~\ref{fig:boundarComp}B.
Fig.~\ref{fig:boundarComp}C has varying curvature along the entire
boundary, even when the walk undergoes a right angle turn. The boundary
inhibits right angles at the these locations because of the contraction at
each step given in Def.~\ref{def::contr}. We can also refine the computed
boundary in Fig.~\ref{fig:boundarComp}C with a Catmull-Clark subdivision
scheme \cite{catmull1978} to recover an ``organic'', i.e. curved, shape.
The Catmull-Clark scheme obtains a fine B-spline approximation of the
boundary, which is valid because no right angles have to be approximated.

\begin{figure*}
\centering 
\includegraphics[scale=0.3]{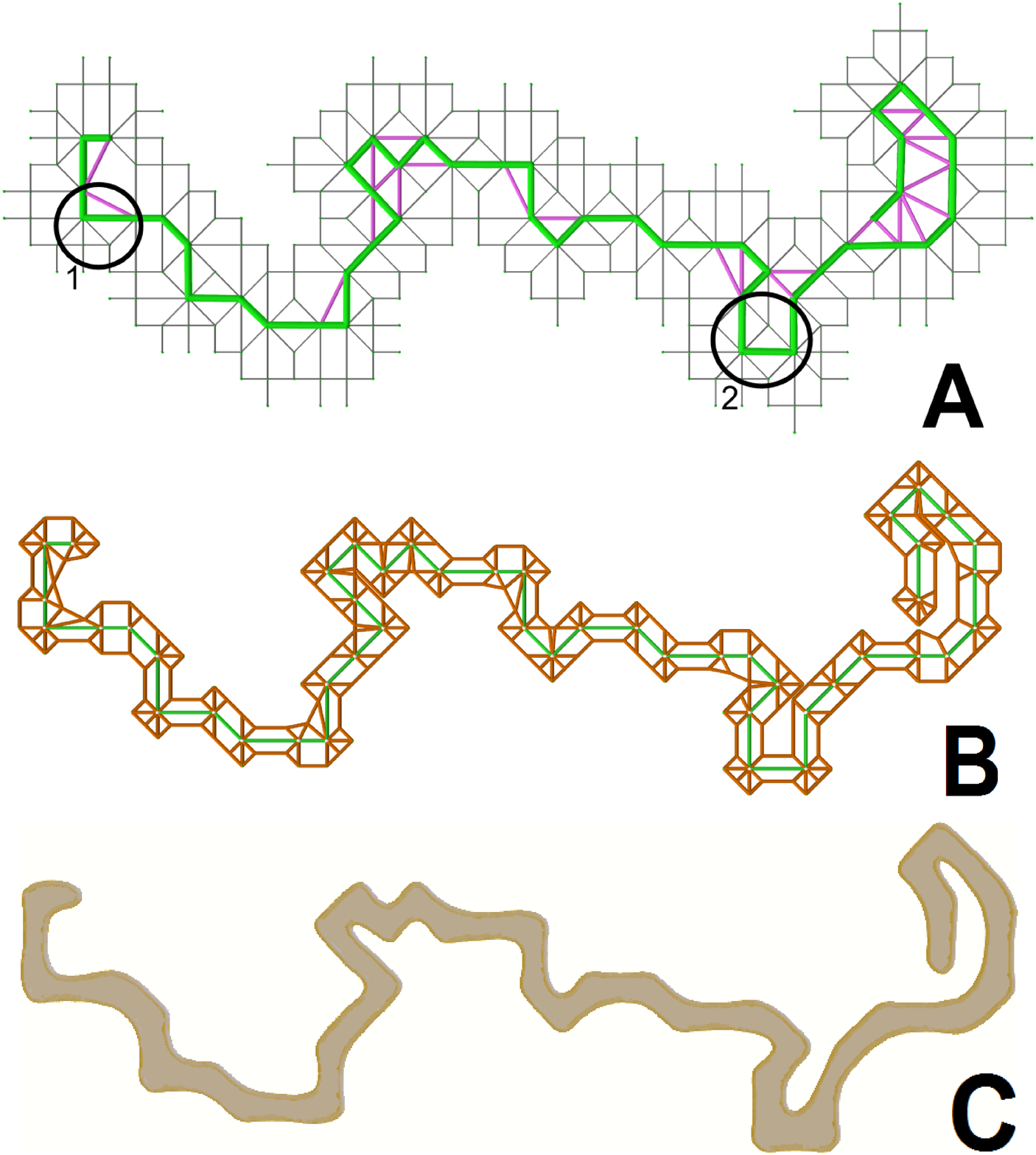}
\caption{ A computed fiber walk example with boundary. (A)The illustrated fiber walk (green) is shown on a grey lattice, with two locations marked where curvature singularities are avoided. (B) The intermediate lattice (orange). (C) The filled boundary (brown) smoothed with a B-Spline in 2D}
\label{fig:boundarComp}
\end{figure*}

\subsection*{Fiber walks define a region of occupied space}

Each step of a fiber walk elongates the fiber and expands its boundary,
which defines a region occupied by each step on the lattice. The fiber
walk contains three edge length classes in 2D of length $1, \sqrt{2}$ and $2$,
which are interpreted as a step forward on uncontracted lattice edges, a
diagonal step resulting from the contraction of the lattice and a step
over previously contracted lattice edges. Each edge length can be
interpreted in terms of the distance the walk must cross through the
expanded region before elongation can occur again. Each forward step
starts at a point inside the actual growing object and has to pass
through an expanded region. See SI Fig. 1-3 for 2D, 3D and 4D for computed
statistics of the observed edge length classes.

\subsection*{Fiber configurations are constrained by spatial expansion}
\label{sub:steplength}
Fiber walks generate self-avoiding edges. Self-avoiding
edges correspond to directions of growth that are inaccessible as a direct
result of the spatial expansion of the fiber.
There are five length classes of self-avoiding edges for a fiber grown on 
a 2D square lattice: $1,\sqrt{2},2,\sqrt{5}$ and $3$. See SI Fig. 1-3 for 2D, 3D and 4D for computed statistics of the observed self-avoiding edge length classes.
Although these length classes vary with lattice-type,
the fact that they are heterogeneous and can be
linked to spatial expansion is generic. In particular,
self-avoiding edges of length $\sqrt{5}$ and $3$ in 2D represent locations
where expanded fibers are closer to each other than needed for expansion. Hence, lateral expansion alter the possible configurations of a resulting fiber.

\begin{figure}[ht]
\centering
\includegraphics[scale=0.38]{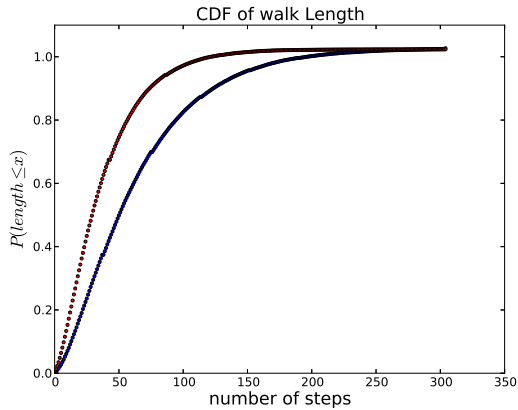}
\caption{Stopping times of SAW and fiber walk. Comparison of the growing SAW (blue) and fiber walk (red) stopping times. The figure shows the computed walk length at termination of 100.000 single walks.}
\label{fig:CDF}
\end{figure}

\subsection*{Spatial expansion inhibits elongation steps of fiber walks}
The stopping time describes the probability of a walk to generate a stopping configuration (compare Def.~\ref{def:fiber}). We give here the stopping times for the 2D growing SAW and fiber walk (Fig.~\ref{fig:CDF}) as the cumulative density function of the walk length (see SI for computation results). For the computation of the 2D stopping times we considered 100,000 walks of a maximum length of 300 steps. For the growing SAW the walks stopped after 50 steps with 50\% chance. In contrast the fiber walk reached the 50\% mark at approximately 29 steps. We found evidence that the two distributions are distinct from each other by performing a two-sample Kolmogorov-Smirnov test ($p<0.001,D=0.25454$). The stopping times allow the conclusion that fiber walks have on average less steps until termination then their dimensionally reduced equivalent. The computation of exact stopping times in 3D and 4D is beyond scope of this paper.

\begin{figure}[ht]
\begin{center}
\includegraphics[scale=0.38]{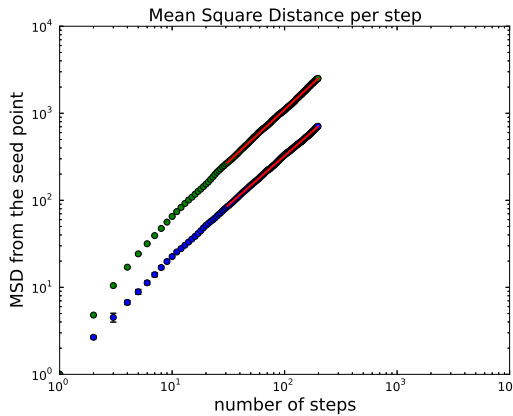} 
\caption{Scaling of the fiber walk and SAW. The summarized overview of the MSD scaling of the growing SAW (blue) and the fiber walk (green) is shown. For both, the fit line in the scaling region is shown in red.}
\label{fig:summary}
\end{center}
\end{figure}

\subsection*{Expansion lets the fiber walk initially reach further}
\label{sub:scale}
From the previous paragraph it follows that we need an indication of how the fiber walk expansion affects the growth of the fiber. We compare the end-to-end distance per step to evaluate how far the growing SAW and the fiber walk are on average away from the origin. The classic property to characterize this behavior is the scaling obtained as the average Euclidean mean-square-displacement (MSD) from the origin as a function of the number of steps on the lattice \cite{madras1993}.
The scaling of the MSD $R(N)$ in relation to the number of steps $N$ taken on the lattice is denoted as:
\begin{equation}
R(N) \propto N^{2\vartheta}\\
\end{equation}
\begin{equation}
2\vartheta= \lim_{N \to \infty}\frac{log(R(N))}{log(N)}
\end{equation}
The scaling exponent $\vartheta$ for a walk is determined as the slope of the best fit line where the function $R(N)$ reaches an asymptotic slope and is called the scaling region. In contrast the region before reaching an asymptotic slope is called the transient region. For example, a standard diffusive process, such as the ordinary random walk, corresponds to $\vartheta=0.5$.
In order to estimate the scaling exponent we computed 1,000 fiber walks and growing SAW's in dimensions 2,3,4 using simple sampling. Trapped walks were restarted by removing the last edge until a walk length of 200 steps was reached to assure that all 1,000 walks represent configurations with in the scaling region. Note that this restarting technique overrides stopping configurations that occur naturally. We selected walks of a certain minimum length to obtain the scaling of the MSD over the number of steps on the lattice. The exponents for all scaling exponents given in this sections are derived from the least-square fit of a line into the scaling region of the averaged MSD values per number of steps. We have chosen a change of less than 0.1 in the asymptotic slope as a threshold to determine the scaling region from 30 steps onwards and used bootstrapping to evaluate the error of the slope. We obtained for $\vartheta$, Fig.~\ref{fig:summary}(left), a value of 0.55$\pm$0.05 and 0.55$\pm$0.06 for the 2D fiber walk and the 2D growing SAW respectively. For dimensions~3~and~4 we obtained 0.53$\pm$0.04 and 0.53$\pm$0.03 for the fiber walk and the 0.52$\pm$0.04 and 0.53$\pm$0.03 for the growing SAW respectively. For computation results in 3D see SI Fig.4 and SI Fig.5 for the 4D computation.
Our main observation is that our fiber walk diverges further away from the origin before reaching the scaling region than the growing SAW, which is visible in the higher absolute MSD value compared to the SAW. This stronger divergence in the transient region suggests that expansion makes fiber walks initially more ballistic because of different step length resulting from the contraction. 

\subsection*{A measure for the overall expanded area}
In alignment with the scaling behaviour of the self-avoiding edges, the number of merges defines the amount of occupied space on the lattice. Nevertheless, the unique characteristic of the fiber walk is the contraction, which defines the boundary describing the spatial expansion of the fiber walk. Here we give the scaling of the number of contractions for the fiber walk as a measure for the for the size of the boundary (See SI Fig. 6-8). In 2D this resulted in a scaling exponent of 1.13, whereas for the 3D and 4D cases a similar exponent of 0.98 and 0.98 was obtained. The similar exponent means that not all edges incident from the side result in a merge. Fig.6-8 in the SI show the actual computation results.

\section*{Discussion}

In this paper we introduced the fiber walk, which is a growing SAW that
includes lateral expansion. The expansion of the fiber walk is modeled as
a local contraction around the last vertex reached on the lattice.  We
have shown that the fiber walk process constitutes a mechanism by which
physical space is reserved, yet does not imply the expansion of the object
into this space.  We have found that the expansion lets a walk diverge
initially further away from the origin before entering the scaling region
and that the fiber walk takes, on average, fewer steps until termination
than the SAW. Self-avoidance as proposed in our model causes encapsulated
regions on the lattice that inhibit further exploration by the fiber unless
a branching mechanism is added.

One benefit of modeling the contraction is that local object thickness
larger than the step size is possible. Although we have shown only the
difference between modeling elongation (alone) and elongation with expansion, our
fiber walk is capable to perform more than one contraction while growing.
A practical variation of the principal fiber walk would be to assign
probabilities to the selection of walk edges. Such probabilities should
allow the simulation of more ballistic walk behaviour.  As far as we are
aware, the connection between self-avoidance and contraction in
SAW-derived models has not been established in the study of random walks
in biology (see \cite{codling2008} for an overview). In particular, the
correlated random walk on a lattice has been studied to explain solely the
elongation and the direction in root development on the example of spruce
trees \cite{renshaw1981}. Recently the importance of spatial expansion for
modeling plants has been highlighted in \cite{prusinkiewicz2010} and
L-Systems might be a suitable formalism to express the fiber walk
mechanisms underlying the emergence of root shapes \cite{leitner2010}.
Simulations of the root growth have been developed that include rich
information on root physiology and biology \cite{dupuy2010,dorlodot2007,lynch1997} and many underlying models have been derived from measured
empirical data \cite{dunbabin2013}.  We see the fiber walk as a
complementary development in the sense that it models a simple but
abstract mechanism with explicit walk geometry and that could be integrated into models that incorporate greater realism.

Future work on the fiber walks will be targeted towards the creation and
simulation of physical networks including the spatial interactions between
several fiber that walk on the same lattice. In our opinion the fiber walk
has the potential to parallel geometrically the apical tip-driven growth
of individual branches in many plants at the apical meristem. Morphogens
within the apical meristem regulate the extension of individual branches
\cite{leyser2011,aloni2006,aloni2004} which may find an analogue in the
tip of the growing fiber. Examples include root growth in plants whose
extension is tip-driven \cite{grieneisen2007,hodge2009,overvoorde2010}.
Above-ground plant and tree structure is tip-driven \cite{sachs2004} and,
again morphogens regulate the branch extension at the apical meristem.
Similarly, the long and branching filamentous structure of fungi (hypha)
extend from the `Spitzenk\"{o}rper' \cite{lew2011}. Also hyphae growth is
simulated via a tip-driven growth emerging from a `vesicle supply center'
within the tip \cite{webster1980}.

Many such networks face the problem of being below the maximum
resolution of sensing technologies in non-laboratory conditions (e.g.
roots in real soil). The fiber walk enables us to develop localized
descriptors or measures that may serve as the basis for models in efforts to
identify adaptive growth rules in complex organisms.  Finally, the current
definition and implement of the fiber walk best describe the growth of
sessile objects.  However, extensions to the model could also include the
possibility of movement of edges within the fiber and further coupling of
physical mechanisms of growth influenced by surface properties. The method
is available within the supporting material of this paper, for download
on the web pages of the authors and on git hub (https://github.com/abucksch/FiberWalk).

\section*{Materials}

The results of this paper were produced with the python program that we
provide together with the paper. This software is suitable for developers
to use as a template to integrate our fiber walk into their software.  An
end user should follow the requirements provided in the ``readme'' file to
execute the program.

\section*{Acknowledgements}
The authors thank Daniel I. Goldman for helpful suggestions on the manuscript.
\bibliography{fiber}

\end{document}